\documentclass[twoside]{article}
\usepackage[a4paper]{geometry}
\usepackage[T1]{fontenc} 
\usepackage{RR}
\usepackage[bookmarks,colorlinks=true,
  citecolor=black,
  linkcolor=black,
  urlcolor=black]{hyperref}

\usepackage{IEEEtrantools}
\usepackage{graphicx}
\usepackage{epstopdf}
\usepackage{varioref}
\usepackage{amsmath}
\usepackage{amssymb}
\usepackage{amsfonts}
\usepackage{amsthm}
\usepackage{mathrsfs}
\usepackage{url}
\usepackage{color}
\usepackage{subfig}

\newcommand{\EE}{\mbox{E}}

\newcommand{\PP}{\mbox{P}}

\theoremstyle{plain}

\theoremstyle{definition}
\newtheorem{definition}{Definition}
\theoremstyle{definition}
\newtheorem*{condition*}{Condition}
\theoremstyle{definition}
\newtheorem{Remarks}{Remark}
\theoremstyle{plain}

\theoremstyle{plain}
\newtheorem{proposition}{Proposition}
\theoremstyle{plain}
\newtheorem{claim}{Claim}

\RRdate{February 2015}

\RRauthor{
Konstantin Avrachenkov\thanks{Inria Sophia Antipolis, France. {Email: \tt k.avrachenkov@sophia.inria.fr}}%
  \and
Natalia M.\ Markovich\thanks{Institute of Control Sciences, Russian Academy of Sciences, Moscow, Russia. {Email: \tt markovic@ipu.rssi.ru}}  
\and \\
Jithin K.\ Sreedharan\thanks{Inria Sophia Antipolis, France. {Email: \tt jithin.sreedharan@inria.fr}}\thanks{Corresponding author}
}

\authorhead{K. Avrachenkov, Natalia M.\ Markovich \& Jithin K.\ Sreedharan}
\RRtitle{Distribution et D\'ependance des Valeurs Extr\^{e}mes dans l'\'{E}chantillonnage des R\'eseaux Complexes}
\RRetitle{Distribution and Dependence of Extremes in Network Sampling Processes}
\titlehead{Distribution and Dependence of Extremes in Network Sampling Processes}
\RRnote{The work of the first and third authors is partly supported by ADR ``Network Science'' of Alcatel-Lucent Inria joint lab. The second author is partly supported by the Russian Foundation for Basic Research, Grant 13-08-00744 A, and Campus France - Russian Embassy bilateral exchange program.}
\RRresume{Nous \'{e}tudions la structure de d\'{e}pendance dans les s\'{e}quences \'{e}chantillonn\'{e}s dans les r\'{e}seaux complexes. Nous consid\'{e}rons des algorithmes randomis\'{e}s pour \'{e}chantillonner les n{\oe}uds et \'{e}tudions les propri\'{e}t\'{e}s extr\'{e}males dans n'importe quelle suite des caract\'{e}ristiques comme degr{\'e}s de n{\oe}uds, nombres d'adeptes ou des revenus des n{\oe}uds des r\'{e}seaux sociaux. Nous faisons l'abstraction de la d{\'e}pendance des extr{\^e}mes par un seul param\`{e}tre qui appara\^{i}t dans la th\'{e}orie des valeurs extr\^{e}mes. Ce param\`{e}tre s'appelle l'Indice Extr\'{e}mal (IE). Dans ce travail, nous \'{e}tudions ce param\`{e}tre analytiquement en utilisant copulas, et empiriquement par estimation. Nous proposons d'utiliser IE comme un param\`{e}tre pour comparer des diff\'{e}rentes proc\'{e}dures d'\'{e}chantillonnage. Comme un exemple, les corr\'{e}lations de degr\'{e} entre les n{\oe}uds voisins sont \'{e}tudi\'{e}es en d\'{e}tail en utilisant les trois techniques d'\'{e}chantillonnage bas\'{e} sur des marches al\'{e}atoires.}
\RRabstract{
We explore the dependence structure in the sampled sequence of large networks. We consider randomized algorithms to sample the nodes and study extremal properties in any associated stationary sequence of characteristics of interest like node degrees, number of followers or income of the nodes in Online Social Networks etc, which satisfy two mixing conditions. Several useful extremes of the sampled sequence like $k$th largest value, clusters of exceedances over a threshold, first hitting time of a large value etc are investigated. We abstract the dependence and the statistics of extremes into a single parameter that appears in Extreme Value Theory, called extremal index (EI). In this work, we derive this parameter analytically and also estimate it empirically. We propose the use of EI as a parameter to compare different sampling procedures. As a specific example, degree correlations between neighboring nodes are studied in detail with three prominent random walks as sampling techniques.
}
\RRmotcle{{\'E}chantillonnage de r{\'e}seau, th{\'e}orie de valeur extr{\^e}me, indice extr{\'e}mal, marche al{\'e}atoire sur un graphe.}
\RRkeyword{Network sampling, extreme value theory, extremal index, random walks on graph.}
\RRprojet{Maestro}
\RCSophia 

\begin{document}
\RRNo{8578}
\makeRR   

\section{Introduction}
Data from real complex networks shows that correlations exist in various forms, for instance the existence of social relationships and interests in social networks. Degree correlations between neighbors, correlations in income, followers of users and number of likes of specific pages in social networks are some examples, to name a few. These kind of correlations have several implications in network structure, for example, degree-degree correlations manifests itself in assortativity or disassortativity of the network \cite{BBV08}.

We consider very large complex networks where it is impractical to have a complete picture a priori. Crawling or sampling techniques are in practice to explore such networks by making use of API calls or HTML scrapping. We look into randomized sampling techniques which generate stationary samples. As an example, random walk based algorithms are in use in many cases because of several advantages offered by them \cite{ART10,Brin1998}.

We focus on the extremal properties in the correlated and stationary sequence of characteristics of interest $X_1,\ldots,X_n$ which is a function of the node sequence, the one actually generated by sampling algorithms. The characteristics of interest, for instance, can be node degrees, node income, number of followers of the node in OSN etc. Among the properties, clusters of exceedances of such sequences over high thresholds are studied in particular. The cluster of exceedances is determined as the consecutive exceedances of $\{X_n\}$ over the threshold $\{u_n\}$ between two consecutive non-exceedances \cite{Ferro,Markovich2014}. It is important to investigate stochastic nature of extremes since it allows us to disseminate advertisement or collect opinions more effectively within the clusters.

The dependence structure of sampled sequence exceeding sufficiently high thresholds is measured using a parameter called extremal index (EI), $\theta$ in Extremal Value Theory. It is defined as follows.
\begin{definition}\cite[p.\ 53]{Leadbetter}
\label{Def-1}
The stationary sequence  $\{X_n\}_{n\ge 1}$, with $F$ as the marginal distribution function and $M_n=\max\{X_1,..., X_n\}$, is said to have the extremal index
$\theta\in[0,1]$ if
for each $0<\tau <\infty$ there is a sequence of real numbers (thresholds) $u_n=u_n(\tau)$ such that
\begin{eqnarray}
\lim_{n\to\infty}n(1-F(u_n))&=&\tau \mbox{ and}  \label{eq:CCDF_condn_theta} \\
\lim_{n\to\infty}\PP\{M_n\le u_n\}&=&e^{-\tau\theta}. \nonumber
\end{eqnarray}
\end{definition}
The maxima $M_n$ is related to EI more clearly as  \cite[p.\ 381]{Beirlant}$^{(\textrm{a})}$
\begin{eqnarray}
\label{eq:max_ei_reln}
\PP\{M_n\le u_n\}&=& F^{n\theta}(u_n)+o(1).
\end{eqnarray}
When $\{X_n\}_{n\ge 1}$ is i.i.d. (for instance uniform independent node sampling), $\theta=1$ and point processes of exceedances over threshold $u_n$ converges weakly to homogeneous Poisson process \cite[Chapter 5]{Beirlant}. But when $0\leq \theta <1$, point processes of exceedances converges weakly to compound Poisson process and this implies that exceedances of high threshold values $u_n$ tend to occur in clusters for dependent data \cite[Chapter 10]{Beirlant}.

EI has many useful interpretations and applications like
\begin{itemize}
\item Finding distribution of order statistics of the sampled sequence. These can be used to find quantiles and predicts the $k$th largest value which arise with a certain probability. Specifically for the distribution of maxima, \eqref{eq:max_ei_reln} is available and the quantile of maxima is proportional to EI. Hence in case of samples with lower EI, lower values of maxima can be expected. When sampled sequence is the sequence of node degrees, these give many useful results.

\item Close relation of extremal index to the distribution and expectation of the size of clusters of exceedances.

\item First hitting time of the sampled sequence to $(u_n, \infty)$ is related to EI. Thus in case of applications where the aim is to detect large values of samples quickly, without actually employing sampling (which might be very costly), we can compare different sampling procedures by EI: smaller EI leads to longer searching of the first hitting time.
\end{itemize}
These interpretations  are explained later in the paper. The network topology determines the stationary distribution of the characteristics of interest under a sampling technique and is reflected on the EI. This indicates that different sampling algorithms may have different EI.

\subsection*{Our contributions}
The main contributions in this work are as follows. We associated Extremal Value Theory of stationary sequences to
sampling of large complex networks and we study the extremal and clustering properties of the sampling process due to correlations. In order to facilitate a painless future study of correlations and clusters of  samples in large networks, we propose to abstract the extremal properties into a single and handy parameter, EI. For any general stationary samples meeting two mixing conditions, we find that knowledge of bivariate distribution or bivariate copula is sufficient to compute EI analytically and thereby deriving many extremal properties. Several useful applications of EI (first hitting time, order statistics and mean cluster size) to analyse large graphs, known only through sampled sequences, are proposed. Degree correlations are explained in detail with a random graph model for which joint degree correlations exist for neighbor nodes. Three different random walk based algorithms that are widely discussed in literature (see \cite{ART10} and the references therein), are then revised for degree state space and EI is calculated when the joint degree correlation is bivariate Pareto distributed. We establish a general lower bound for EI in PageRank processes irrespective of the degree correlation model. Finally two estimation techniques of EI are provided and EI is numerically computed for a synthetic graph with neighbour degrees correlated and for two real networks (Enron email network and DBLP network). 

The paper is organized as follows. In Section \ref{sec:calcn_EI}, methods to derive EI are presented. Section \ref{sec:deg_corlns} considers the case of degree correlations. In Section \ref{sec:graph} the graph model and correlated graph generation technique are presented. Section \ref{sec:desc_rand_walks} explains the different types of random walks studied and derives associated transition kernels and joint degree distributions. EI is calculated for different sampling techniques later in Section \ref{sec:deg_calcn_ei}. In Section \ref{sec:applcn_ei} we provide several applications of extremal index in graph sampling techniques. In Section \ref{sec:simulations} we estimate extremal index and perform numerical comparisons. Finally Section \ref{sec:conclu} concludes the paper.

A shorter version of this submission has been appeared in \cite{Jithin_2014}.

\section{Calculation of Extremal Index (EI)}
\label{sec:calcn_EI}
We consider networks represented by an undirected graph $G$ with $N$ vertices and $M$ edges. Since the networks under consideration
are
 huge, we assume it is impossible to describe them completely, i.e., no adjacency matrix beforehand. Assume any randomized sampling procedure is employed and let the sampled sequence $\{X_i\}$ be any general sequence.

This section explains a way to calculate extremal index from the bivariate distribution if the sampled sequence admits two mixing conditions.
\begin{condition*}[$D(u_n)$]
\label{cond:d}
\begin{multline*}
\hspace*{-0.4 cm}\Big|\PP(X_{i_1}\leq u_n,\ldots,X_{i_p}\leq u_n,X_{j_1}\leq u_n,\ldots,X_{j_q}\leq u_n )\\
-\PP(X_{i_1}\leq u_n,\ldots,X_{i_p}\leq u_n)\PP(X_{j_1}\leq u_n,\ldots,X_{j_q}\leq u_n )\Big|\leq \alpha_{n,l_n},
\end{multline*}
where $\alpha_{n,l_n} \to 0$ for some sequence $l_n=o(n)$ as $n \to \infty$, for any integers $i_1\leq \ldots <i_p<j_1<\ldots\leq j_q$ with $j_1-i_p>l_n$.
\end{condition*}

\begin{condition*}[$D''(u_n)$]
\label{cond:d''}
\begin{equation*}
\lim_{n \to \infty} \Big \{\sum_{j=2}^{r_n}\PP(X_j \leq u_n < X_{j+1}|X_1>u_n) \Big\}= 0,
\end{equation*}
where $(n/r_n)\alpha_{n,l_n} \to 0$ and $l_n/r_n \to 0$ with $\alpha_{n,l_n}$, $l_n$ as in Condition $D(u_n)$ and $r_n$ as $o(n)$.
\end{condition*}
%

Let $C(u,v)$ is the bivariate Copula \cite{Nelsen2007} ($[0,1]^2 \to [0,1]$) and $C'$ is its G\^{a}teaux derivative along the direction $(1,1)$. Using Sklar's theorem \cite[p.\ 18]{Nelsen2007}, with $F$ as the marginal stationary distribution function of the sampling process,
\[C(u,u)=\PP(X_1\leq F^{-1}(u), X_2\leq F^{-1}(u)).\]
$F^{-1}$ denotes the inverse function of $F$
. This representation is unique if the stationary distribution $F(x)$ is continuous. 

\begin{proposition}
\label{prop:theta}
If the sampled sequence is stationary and satisfies conditions $D(u_n)$ and $D''(u_n)$, then extremal index is given by
\begin{equation}
\label{eq:theta_my_expn}
\theta=C'(1,1)-1,
\end{equation}
and $0\leq \theta \leq 1$.
\end{proposition}
\begin{proof}
From \cite{Leadbetter_1989}, for the stationary sequence $\{X_n\}$ with Conditions $D(u_n$ and $D''(u_n)$, $\theta = \lim_{n \to \infty} \PP(X_2\leq u_n |X_1>u_n)$. Then
\begin{eqnarray}
\theta &=&\lim_{n \to \infty} \frac{\PP(X_2 \leq u_n,X_1>u_n)}{\PP(X_1>u_n)}\nonumber \\
&=&\lim_{n \to \infty} \frac{\PP(X_2 \leq u_n)-\PP(X_1 \leq u_n, X_2 \leq u_n)}{\PP(X_1>u_n)} \nonumber \\
&=&\lim_{n \to \infty} \frac{\PP(X_2 \leq u_n)-C\big (\PP(X_1 \leq u_n), \PP(X_2 \leq u_n)\big)}{1-\PP(X_1 \leq u_n)} \nonumber\\
&=&\lim_{x \to 1} \frac{x-C(x,x)}{1-x} \nonumber\\
&=& C'(1,1)-1. \nonumber
\end{eqnarray}
The existence of EI in $[0,1]$ is evident from the definition used in this proof.
\end{proof}

\begin{Remarks}
The condition $D''(u_n)$ can be made weaker to $D^{(k)}(u_n)$ presented in \cite{Chernick_1991},
\[\lim_{n \to \infty} n\PP\left(X_1>u_n\geq \max_{2\leq i \leq k} X_i, \max_{k+1\leq j \leq r_n} X_j >u_n \right)=0,\]
where $r_n$ is defined as in $D''(u_n)$. For the stationary sequence $D^{(2)}(u_n)=D''(u_n)$. If we assume $D^{(k)}$ is satisfied for some $k\geq 2$ along with $D(u_n)$, then following the proof of Proposition \ref{prop:theta}, EI can be derived as \[\theta=C'_k(1)-C'_{k-1}(1),\] where $C_k(x)$ represents the copula of
$k$-dimensional vector $(x_1,\ldots,x_k)$, $C_k(x_1,\ldots,x_k)$ with $x=x_1\ldots=x_k$
and $C_{k-1}$ is its $(k-1)$th marginal, $C_{k-1}(x)=C_{k-1}(x_1,\ldots,x_{k-1},1)$ with $x=x_1\ldots=x_{k-1}$.
\end{Remarks}

In some cases it is easy to handle with the joint tail distribution. Survival Copula $\widehat{C}(\cdot,\cdot)$ which corresponds to
 \[\PP(X_1> x, X_2> x)=\widehat{C}(\overline{F}(x), \overline{F}(x)),\]
with $\overline{F}(x)=1-F(x)$, can also be used to calculate $\theta$. It is related to Copula as $\widehat{C}(u,u)=C(1-u,1-u)+2u-1$ \cite[p.\ 32]{Nelsen2007}. Hence $\theta=C'(1,1)-1=1-\widehat{C}'(0,0)$.

Lower tail dependence function of survival copula is defined as \cite{Weng2012}
\[\lambda(u_1,u_2)=\lim_{t\to 0^{+}}\frac{\widehat{C}(tu_1,tu_2)}{t}.\] Hence $\widehat{C}'(0,0)=\lambda(1,1)$. $\lambda$ can be calculated for different copula families. In particular, if $\widehat{C}$ is a bivariate Archimedean copula, then it can be represented as, $\widehat{C}(u_1,u_2)=\psi(\psi^{-1}(u_1)+\psi^{-1}(u_2))$, where $\psi$ is the generator function and $\psi^{-1}$ is its inverse with $\psi:[0,\infty]\to[0,1]$ meeting several other conditions. If $\psi$ is a regularly varying distribution with index $-\beta$,
$\beta>0$,
 then $\lambda(x_1,x_2)=(x_1^{-\beta^{-1}}+x_2^{-\beta^{-1}})^{-\beta}$ and $(X_1,X_2)$ has a multivariate regularly varying distribution \cite{Weng2012}. Therefore, for Archimedean copula family, EI is given by
\begin{equation}
\label{eq:theta_archim}
\theta=1-1/2^{\beta}.
\end{equation}
As an example, bivariate Pareto distribution of the form $\PP(X_1> x_1, X_2> x_2)=(1+x_1+x_2)^{-\gamma}$, $\gamma>0$ has Arhimedean copula with generator function $\psi(x)=(1+x)^{-\gamma}$. This gives $\theta=1-1/2^{\gamma}$. Bivariate exponential distribution of the form
\[\PP(X_1> x_1, X_2> x_2)=1-e^{-x_1}-e^{-x_2}+e^{-(x_1+x_2+\eta x_1x_2)},\] $0\leq \eta \leq 1,$ also admits Archimedian copula.

\subsection{Check of conditions $D(u_n)$ and $D''(u_n)$}
\label{subsec:check_condns}
If the sampling technique is assumed to be based on a Markov chain and consider the sampled sequence as measurable functions of stationary Markov samples, then such a sequence is stationary and \cite{Brien_1987} proved that
another mixing condition $AIM(u_n)$ which implies $D(u_n)$ is satisfied. 

Condition $D''(u_n)$
allows clusters with consecutive exceedances and eliminates the possibility of
clusters with upcrossing of the threshold $u_n$ (${X_i\leq u_n <X_{i+1}}$).
Hence in those cases, where it is tedious to check the condition $D''(u_n)$ theoretically, we can use numerical procedures to measure ratio of number of consecutive exceedances to number of exceedances and the ratio of number of upcrossings to number of consecutive exceedances in small intervals. Such an example is provided in Section \ref{sec:deg_calcn_ei}.

\begin{Remarks}
\label{ferreira_comment}
The EI is derived in \cite{FerFer} to the same expression in \eqref{eq:theta_my_expn}. But \cite{FerFer} assumes $\{X_n\}$ is sampled from a first order Markov chain. This condition is much stricter than $D(u_n)$ and $D''(u_n)$ which we used to derive \eqref{eq:theta_my_expn}. For instance, degrees of the node samples obtained from a Markov chain based sampling, mostly not form a Markov chain as node-degree relation is not one-one while $D(u_n)$ is agreed for such a case and $D''(u_n)$ can get satisfied, see Section \ref{sec:deg_calcn_ei} for an example.
%

\end{Remarks}
\section{Degree correlations}
\label{sec:deg_corlns}
The techniques established in Section \ref{sec:calcn_EI} are very general, applicable to any sampling techniques and any sequence of samples which satisfy certain conditions. In this section we illustrate the calculation of extremal index for correlations among degrees. We introduce different sampling techniques through this section though they can be used in case of any general correlations. We denote the sampled sequence $\{X_i\}$ as $\{D_i \}$ in this section.
\subsection{Description of the model}
\label{sec:graph}
We take into account correlation in degrees between neighbor nodes. The dependence structure in the graph is described by the joint degree-degree probability density function $f(d_1,d_2)$ with $d_1$ and $d_2$ indicating the degrees of adjacent nodes or equivalently by the corresponding tail distribution function $\overline{F}(d_1,d_2)=\PP(D_1 \ge d_1, D_2 \ge d_2)$ with $D_1$ and $D_2$ representing the corresponding degree random variables (see e.g., \cite{BBV08, BPV03,GDM08}).

The probability that a randomly chosen edge has the end vertices with degrees $d_1 \leq d \leq d_1+\Delta(d_1)$ and $d_2 \leq d \leq d_2+\Delta(d_2)$ is $(2-\delta_{d_1d_2})f(d_1,d_2)\Delta(d_1)\Delta(d_2)$. Here $\delta_{d_1d_2}=1$ if $d_1=d_2$, zero otherwise. The multiplying factor $2$ appear on the above expression when $d_1 \neq d_2$ because of the symmetry in $f(d_1,d_2)$, $f(d_1,d_2)=f(d_2,d_1)$ due to the undirected nature of the underlying graph, and the fact that both $f(d_1,d_2$ and $f(d_2,d_1)$ contribute to the edge probability under consideration.

The degree density $f_d(d_1)$ can be calculated from the marginal of $f(d_1,d_2)$ as
\begin{equation}
\label{eq:marginal_jt-degree}
f(d_1)=\int_{d_2} f(d_1,d_2)d(d_2) \approx \frac{d_1 f_d(d_1)}{\EE[D]},
\end{equation}
where $\EE[D]$ denotes the mean node degree,
\[\EE[D]=\left[\int \int \left( \frac{f(d_1,d_2)}{d_1} \right) d(d_1)d(d_2)\right]^{-1}.\]
$f(.)$ can be interpreted as the degree density of a vertex reached by following a randomly chosen edge. The approximation for $f(d_1)$ is obtained as follows: in the R.H.S.\ of (\ref{eq:marginal_jt-degree}), roughly, $d_1 f_d(d_1)N$ is the number of half edges from nodes with degree around $d_1$ and $\EE[D] N$ is the total number of half edges.

From the above description, it can be noted that the knowledge of $f(d_1,d_2)$ is sufficient to describe this random graph model and for its generation.

Most of the results in this paper are derived assuming continuous probability distributions for $f(d_1,d_2)$ and $f_d(d_1)$ because an easy and unique way to calculate extremal index exists for continuous distributions in our setup (more details in Section \ref{sec:calcn_EI}). Also the extremal index might not exist for many discrete valued distributions \cite{Leadbetter}.
\subsubsection{Random graph generation}
\label{subsubsec:graph-genern}
A random graph bivariate joint degree-degree correlation distribution can be generated as follows (\cite{Newman2002}).
\begin{enumerate}
\item Degree sequence is generated according to the degree distribution, $f_d(d)=\frac{f(d)E[D]}{d}$
\item An uncorrelated random graph is generated with the generated degree sequence using configuration model (\cite{BBV08})
\item Metropolis dynamics is applied now on the generated graph: choose two edges randomly (denoted by the vertex pairs $(v_1,w_1)$ and $(v_2,w_2)$) and measure the degrees, $(j_1,k_1)$ and $(j_2,k_2)$ correspond to these vertex pairs. Generated a random number, $y$, according to uniform distribution in $[0,1]$. If $y\leq \min(1,(f(j_1,j_2)f(k_1,k_2))/(f(j_1,k_1)f(j_2,k_2)))$, then remove the selected edges and construct news ones as $(v_1,v_2)$ and $(w_1,w_2)$. Otherwise keep the selected edges intact. This dynamics will generate the required joint degree-degree distribution. Run Metropolis dynamics well enough to mix the network.
\end{enumerate}

\subsection{Description of random walks}
\label{sec:desc_rand_walks}
In this section, we explain three different random walk based algorithms for exploring the network. They have been extensively studied in previous works \cite{ART10,Brin1998, L93} where they are formulated with vertex set as the state space of the underlying Markov chain on graph. The walker in these algorithms, after reaching each node, moves to another node randomly by following the transition kernel of the Markov chain. But since the interest in the present work is in the degree sequence, rather than node sequence, and its extremal properties, we take degree set as the state space and find appropriate transition kernels. We use ${f}_{\mathscr{X}}$ and $\PP_{\mathscr{X}}$to represent the probability density function and probability measure under the algorithm $\mathscr{X}$ with the exception that $f_d$ represents the probability density function of degrees.

\subsubsection{Random Walk (RW)}
In a random walk, the next node to visit is chosen uniformly among the neighbors of the current node. From (\ref{eq:marginal_jt-degree}) we approximate the standard random walk on degree state space by the following transition kernel, conditional density function that the present node has degree $d_t$ and the next node is with degree $d_{t+1}$,
\begin{equation}
\label{eq:RWkernel}
f_{RW}(d_{t+1}|d_{t}) \approx \frac{\EE[D] f(d_t,d_{t+1})}{d_t f_d(d_t)}.
\end{equation}
This approximation is obtained as follows: given the present node has degree $d_t$, $1/d_t$ is the probability of selecting a neighbor uniformly and rest of the terms in R.H.S.\ represent the mean number of neighbors with degree around $d_{t+1}$. When $d_t \neq d_{t+1}$, $\frac{\EE[D]N}{2}(2f(d_t,d_{t+1}))$ is the mean number of edges between degrees about $d_t$ and $d_{t+1}$ and $f_d(d_t)N$ is the mean number of nodes with degrees about $d_t$, and thus their ratio represents such a mean number of edges per node with degree about $d_t$, i.e., mean number of neighbors with degree about $d_{t+1}$.  The probability of occurring the other case, $d_t=d_{t+1}$, is zero as the degrees are assumed to follow a continuous distribution.

If the standard random walk on the vertex set is in the stationary regime, its stationary distribution (probability of staying at a particular vertex $i$) is proportional to the degree (see e.g., \cite{L93}) and is given by $d_i/2M$. Then in the standard random walk on degree set, the stationary distribution of staying at any node with degree around $d_1$ can be approximated as $Nf_d(d_1)\left(d_1/2M|\right)$. Thus
\begin{equation*}
f_{RW}(d_1) \approx \frac{d_1}{\EE[D]}  f_d(d_1),
\end{equation*}
Then, the joint density of the standard random walk is $f_{RW}(d_{t+1}, d_{t}) \approx f(d_t,d_{t+1}).$\\
\subsubsection*{Check of the approximation}
\label{subsubsec:RW_trans_comp}
We provide comparison of simulated values and theoretical values of transition kernel of RW in Figure \ref{fig:trans_kernel_RW}. The bivariate Pareto model is assumed for the joint degree-degree tail function of the graph,
\begin{equation}
\label{eq:biPareto}
\bar{F}(d_1,d_2) = \left(1+\frac{d_1-\mu}{\sigma}+\frac{d_2-\mu}{\sigma}\right)^{-\gamma},
\end{equation}
where $\sigma$, $\mu$ and $\gamma$ are positive values. In the figure, $N$ number of nodes is 5,000. $\mu=10$, $\gamma=1.2$ and $\sigma=15$. These choices of parameters provides $E[D]=21.0052$.  At each instant Metropolis dynamics will choose two edges and it has run 200,000 times (provides sufficient mixing). The figure shows satisfactory fitting of the approximation.

\begin{figure}[!htb]
\centering
\includegraphics[trim = 16mm 0 19mm 5mm, clip=true, scale=0.23]{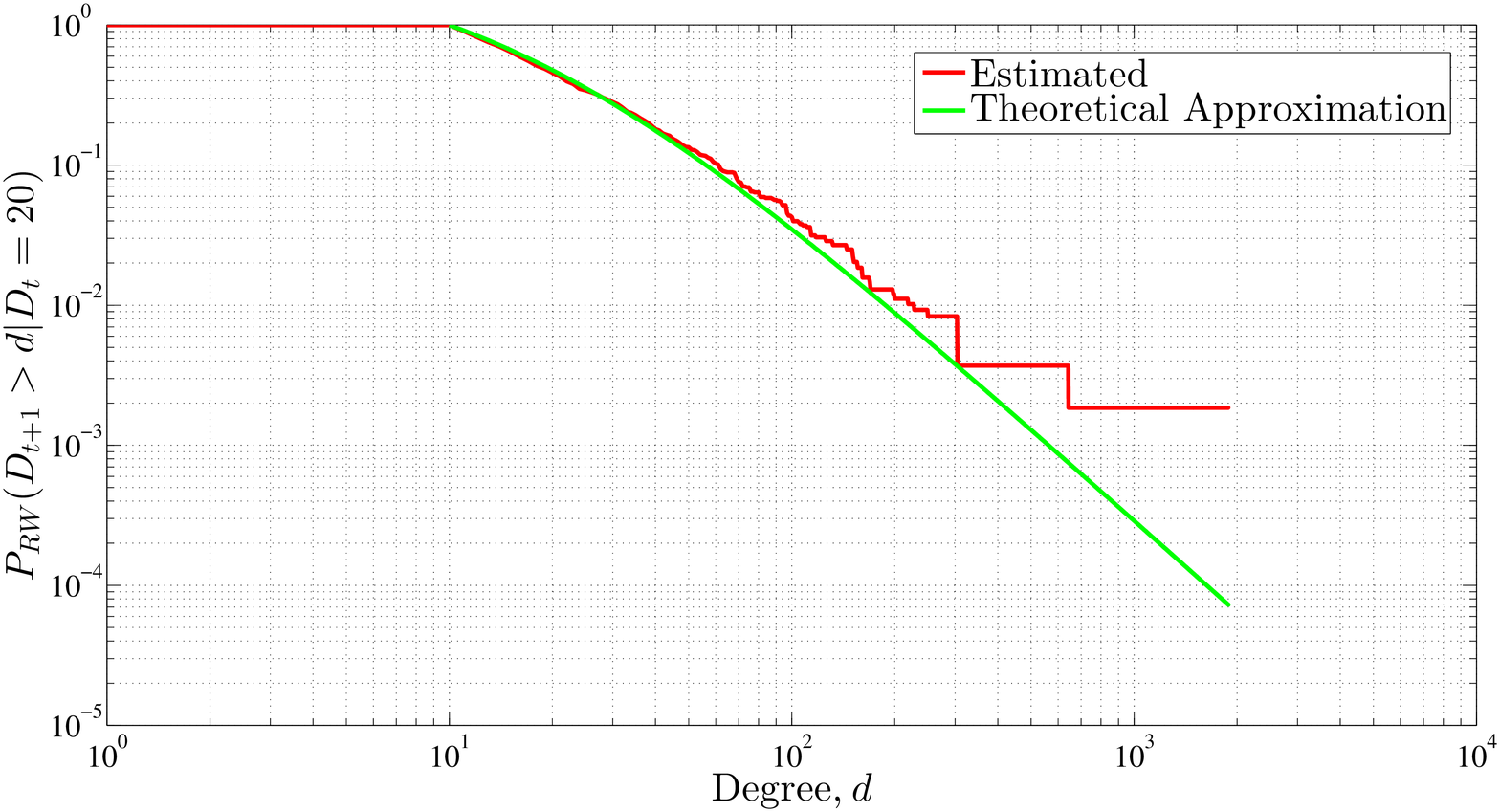}
\caption{Transition kernel comparison}
\label{fig:trans_kernel_RW}
\end{figure}

\subsubsection{PageRank (PR)}
\label{sec:PR}
PageRank is a modification of the random walk which with a fixed probability $1-c$ samples a random node with uniform distribution and with a probability $c$, it follows the random walk transition  \cite{Brin1998}. Its evolution on degree state space can be described as follows:
\begin{IEEEeqnarray}{rCl}
\label{eq:PRkernel}
f_{PR}(d_{t+1}|d_{t})&\approx & c\; f_{RW}(d_{t+1}|d_{t}) + (1-c) \frac{1}{N} N f_d(d_{t+1})\nonumber\\
&\approx & c\; f_{RW}(d_{t+1}|d_{t}) + (1-c) f_d(d_{t+1})
\end{IEEEeqnarray}
Here the $1/N$ corresponds to the uniform sampling on vertex set and $\frac{1}{N} N f_d(d_{t+1})$ indicates the net probability of jumping to all the nodes with degree around $d_{t+1}$.
\subsubsection*{Check of the approximation} We provide a consistency check of the approximation derived for transition kernel by studying tail behavior of degree distribution and PageRank distribution. It is known that under some strict conditions, for a directed graph, PageRank and Indegree have same tail exponents \cite{Litvak}. In our formulation in terms of degrees, for \textit{uncorrelated} and undirected graph, PageRank for a given degree $d$, $PR(d)$, can be approximated from the basic definition as,
\[PR(d)=f_{RW}(d)\approx c\;f_{RW}(d)+(1-c)\;f_d(d).\]
This is a deterministic quantity.
We are interested in the distribution of the random variable $PR(D)$, PageRank of a randomly choosen degree class $D$. PageRank $PR(d)$ is also the long term proportion or probability that PageRank process ends in a degree class with degree $d$. This can be scaled suitably to provide a rank-type information. Its tail distribution is
\begin{eqnarray}
P(PR(D)> x)=P\left(c.f_{RW}(D)+(1-c).f_d(D) > x \right), \nonumber
\end{eqnarray}
where $D \sim f_d(.)$. The PageRank of any vertex inside the degree class $d$ is $PR(d)/(Nf_d(d))$. The distribution of Page Rank of a randomly chosen vertex $i$, $P(PR(i)>x)$ after appropriate scaling for comparison with degree distribution is $P(N.PR(i)>\hat{d})$, where $\hat{d}=Nx$. Now
\begin{eqnarray}
P(N.PR(i)>\hat{d})&=&P\left(N \frac{PR(D)}{Nf_d(D)}>\hat{d}\right) \nonumber \\
&=&P\left(D>\frac{E[D]}{c}\left[\hat{d}-(1-c)\right] \right). \nonumber
\end{eqnarray}
This of the form $P(D>A\hat{d}+B)$ with $A$ and $B$ as appropriate constants and hence will have the same exponent of degree distribution tail when the graph is \textit{uncorrelated}.

There is no convenient expression for the stationary distribution of PageRank, to the best of our knowledge, and it is difficult to come up with an easy to handle expression for the joint distribution. Therefore, along with other advantages, we consider another modification of the standard random walk.
\subsubsection{Random Walk with Jumps (RWJ)}
\label{RWJ}
RW sampling leads to many practical issues like the possibility to get stuck in a disconnected component, biased estimators etc. RWJ overcomes such problems (\cite{ART10}).

In this algorithm we follow random walk on  a modified graph which is a superposition of the given graph and complete graph on same vertex set of the given graph with weight $\alpha/N$ on each edge,  $\alpha\in[0,\infty]$ being a design parameter (\cite{ART10}). The algorithm can be shown to be equivalent to select $c=\alpha/(d_t+\alpha)$ in the PageRank algorithm, where $d_t$ is the degree of the present node. The larger the node's degree, less likely is the artificial jump of the process. This modification makes the underlying Markov chain time reversible, significantly reduces mixing time, improves estimation error and leads to a closed form expression for stationary distribution.

The transition kernel on degree set, following PageRank kernel, is
\begin{IEEEeqnarray}{rCl}
f_{RWJ}(d_{t+1}|d_{t}) &\approx &\frac{d_t}{d_t+\alpha}f_{RW}(d_{t+1}|d_{t})+\frac{\alpha}{d_t+\alpha}f_d(d_{t+1})\nonumber \\
&=&\frac{\EE[D]f(d_t,d_{t+1})+\alpha f_d(d_t)f_d(d_{t+1})}{(d_t+\alpha)f_d(d_t)}.\nonumber \IEEEeqnarraynumspace
\end{IEEEeqnarray}

The stationary distribution for node $i$ (on the vertex set) is $(d_i+\alpha)/(2M+N\alpha)$ and the equivalent stationary probability density function on degree set by collecting all the nodes with same degree is
\begin{eqnarray}
\label{eq:JPsta}
f_{RWJ}(d_1) &\approx &\left( \frac{d_1+\alpha}{2M+N\alpha}\right)Nf_d(d_1) \nonumber\\
&=&\frac{(d_1+\alpha)f_d(d_1)}{\EE[D]+\alpha},
\end{eqnarray}
since $2M/N=\EE[D]$. The stationarity of the $f_{RWJ}(d_1$ can be verified by plugging the obtained expression in the stationarity condition of the Markov Chains. We have
\begin{IEEEeqnarray*}{rCl}
{f_{RWJ}(d_1)}
&=&\int f_{RWJ}(d_1|d_2)f_{RWJ}(d_2) d({d_2})\nonumber \\
&\approx &\int \frac{\EE[D]f(d_1,d_2)+\alpha f_d(d_1)f_d(d_2)}{(d_2+\alpha)f_d(d_2)}\;
\frac{(d_2+\alpha)f_d(d_2)}{\EE[D]+\alpha} d({d_2}) \nonumber \IEEEeqnarraynumspace \\
\label{eq:RWJPmarginal}
&\approx &\frac{(d_1+ \alpha)  f_d(d_1)}{\EE[D]+\alpha},
\end{IEEEeqnarray*}
where (\ref{eq:marginal_jt-degree}) has been applied.
Then, the joint density function for the random walk with jumps has the following form
\begin{IEEEeqnarray*}{rCl}
f_{RWJ}(d_{t+1}, d_{t})\approx \frac{\EE[D]f(d_{t+1},d_t)+\alpha f_d(d_{t+1})f_d(d_t)}{\EE[D]+\alpha}. \IEEEeqnarraynumspace
\end{IEEEeqnarray*}
Moreover the associated tail distribution has a simple form,
\begin{equation}
\label{eq:tail_RWJ_joint-distbn}
f_{RWJ}(D_{t+1} > d_{t+1}, D_{t} > d_{t}) \approx \frac{\EE[D]\overline{F}(d_{t+1},d_t)+\alpha \overline{F}_d(d_{t+1})\overline{F}_d(d_t)}{\EE[D]+\alpha}.
\end{equation}
\begin{Remarks}
Characterizing Markov chain based sampling in terms of degree transition has some advantages,
\begin{itemize}
\item In the different random walk algorithms considered on vertex set, all the nodes with same degree have same stationary distribution. This also implies that it is more natural to formulate the random walk transition in terms of degree.
\item Degree uncorrelations in the underlying graph is directly reflected in the joint distribution of the studied sampling techniques. For uncorrelated networks, $f_{RW}(d_1,d_2)=f_{RW}(d_1)f_{RW}(d_2)$, $f_{PR}(d_1,d_2)=f_{PR}(d_1)f_{PR}(d_2)$ and $f_{RWJ}(d_1,d_2)=f_{RWJ}(d_1)f_{RWJ}(d_2)$.
\end{itemize}
\end{Remarks}

\subsection{Extremal Index for bivariate Pareto Degree Correlation}
\label{sec:deg_calcn_ei}
As explained in the Introduction section, extremal index is an important parameter in characterizing dependence and extremal properties in a stationary sequence. We assume that we have waited sufficiently long that the underlying Markov chain of  the three different graph sampling algorithms are in stationary regime now. Here we derive EI of RW and RWJ for the model with degree correlation among neighbours as bivariate Pareto \eqref{eq:biPareto}.

The two mixing conditions $D(u_n)$ and $D''(u_n)$ introduced in Section \ref{sec:calcn_EI} are needed for our EI analysis. Condition $D(u_n)$ is satisfied as explained in Section \ref{subsec:check_condns}. An empirical evaluation of $D''(u_n)$ is provided in Section \ref{subsub:check_d2_emp}. 

\subsubsection{EI for Random Walk sampling}
\label{subsec:theta_ex_RW}
We use the expression for EI given in Proposition \ref{prop:theta}. As $f_{RW}(x,y)$ is same as $f(x,y)$, we have,
\begin{eqnarray}
\widehat{C}(u,u)&=&\PP(D_1>\bar{F}^{-1}(u),D_2>\bar{F}^{-1}(u))\nonumber \\
&=&\left(1+2(u^{-1/\gamma}-1) \right)^{-\gamma} \nonumber \\
\widehat{C}'(u,u)&=& 2(2-u^{1/\gamma})^{-(\gamma+1)}. \nonumber
\end{eqnarray}
Thus $\theta=1-\widehat{C}'(0,0)=1-1/2^{\gamma}$. For $\gamma=1$ we get $\theta=1/2$. In this case, we can also use expression given in \eqref{eq:theta_archim}.

\subsubsection{EI for Random Walk with Jumps sampling}
Although it is possible to derive EI as in RW case above, we provide an alternative way to avoid the calculation of tail distribution of degrees and inverse of RWJ marginal (with respect to the bivariate Pareto degree correlation).

Under the assumption of $D''$, 
\begin{equation}
\label{eq:proof_RWJ_2pareto}
\theta = \lim_{n \to \infty}
\frac{P(D_2 \le u_n, D_1>u_n)}{P(D_1 > u_n)}
= \lim_{n \to \infty}
\frac{P(D_1 \ge u_n)-P(D_2 \ge u_n, D_1 \ge u_n)}{P(D_1 > u_n)}
\end{equation}

Now using the condition \eqref{eq:CCDF_condn_theta} on the marginal and joint tail distribution of RWJ \eqref{eq:tail_RWJ_joint-distbn}, we can write$^{(\textrm{b})}$

\begin{eqnarray}
\lefteqn{\frac{P(D_1 \ge u_n)-P(D_2 \ge u_n, D_1 \ge u_n)}{P(D_1 > u_n)}} \nonumber\\
&=&\frac{\tau/n+o(1/n) - \frac{E[D]}{E[D]+\alpha}P_{RW}(D_2 \ge u_n, D_1 \ge u_n) - \frac{\alpha}{E[D]+\alpha}O(\tau/n)O(\tau/n) }{\tau/n+o(1/n)} \nonumber
\end{eqnarray}

The asymptotics in the last term of the numerator is due to the following:
\[\overline{F}_{RWJ}(u_n)=\frac{E[D]}{E[D]+\alpha}\overline{F}(u_n)+\frac{\alpha}{E[D]+\alpha} \overline{F}_d(u_n)=\tau/n+o(1/n),\]
and hence $\overline{F}_d(u_n)=O(\tau/n)$. Therefore \eqref{eq:proof_RWJ_2pareto} becomes
\begin{equation*}
\theta=1-\frac{E[D]}{E[D]+\alpha}\lim_{n \to \infty}P_{RW}(D_2 \ge u_n, D_1 \ge u_n)n/\tau
\end{equation*}

In the case of the bivariate Pareto distribution \eqref{eq:biPareto}, we obtain
\begin{equation}
\label{eq:EI-RWJ_2pareto}
\theta = 1 - \frac{E[D]}{E[D]+\alpha} 2^{-\gamma}
\end{equation}

\subsection{Lower bound of EI of the PageRank}
We obtain the following lower bound for EI in the PageRank processes.
\begin{proposition}
\label{prop:PR}
For the PageRank process on degree state space irrespective of the degree correlation structure in the underlying graph, the extremal index  
\[\theta \geq (1-c).\]
\end{proposition}
\begin{proof}
From \cite{Brien_1987}, the following representation of EI holds for degree sequence, 
\begin{equation}
\label{eq:Brien_EI_repsn}
\lim_{n\to\infty}\PP\{M_{1,p_n}\le u_n|D_{1} >u_n\}= \theta,
\end{equation}
where $\{p_n\}$ is an increasing sequence of positive integers, $p_n=o(n)$ as $n\to\infty$ and $M_{1,p_n}=\max\{D_2,..., D_{p_n}\}$. Let $\mathcal{A}$ be the event that the node corresponding to $D_2$ is selected uniformly among all the nodes, not following random walk from the node for $D_1$. Then $\PP_{PR}(\mathcal{A})=1-c$. Now, with (\ref{eq:PRkernel}), 
\begin{IEEEeqnarray}{rCl}
{\PP _{PR}(M_{1,p_n}\le u_n|D_{1}>u_n)}
& \geq & \PP_{PR} (M_{1,p_n}\le u_n, \mathcal{A}|D_{1}>u_n) \nonumber\\
&=& \PP_{PR}(\mathcal{A}|D_1>u_n) \PP_{PR}(M_{1,p_n}\le u_n|\mathcal{A},D_{1}>u_n) \nonumber \\
&\stackrel{(i)}=&(1-c)\PP_{PR}(M_{1,p_n}\le u_n), \nonumber \\
&\stackrel{(ii)}=&(1-c)\PP_{PR}^{(p_n-1)\theta}({D_{1}\le u_n})+o(1)\nonumber\\
&\geq &(1-c)\PP_{PR}^{(p_n-1)}({D_{1}\le u_n})+o(1) \nonumber \\ 
&\stackrel{(iii)}\sim & (1-c)(1-\tau/n)^{p_n-1}, \label{eq:proof_prop2_2}
\end{IEEEeqnarray}
where $\{p_n\}$ is the same sequence as in (\ref{eq:Brien_EI_repsn}) and $(i)$ follows mainly from the observation that conditioned on $\mathcal{A}$, $\{M_{1,p_n}\le u_n \}$ is independent of $\{D_{1}>u_n\}$, $(ii)$ and $(iii)$ result from the approximations in \eqref{eq:max_ei_reln} and \eqref{eq:CCDF_condn_theta} respectively.

Assuming $p_n-1=n^{1/2}$ and since $(1-\tau/n)^{p_n-1}\sim e^{-\tau/\sqrt n}\rightarrow 1$ as $n\to\infty$, from \eqref{eq:Brien_EI_repsn} and \eqref{eq:proof_prop2_2},
\begin{equation*}
\label{14}\theta \ge 1-c.
\end{equation*}

The PageRank transition kernel \eqref{eq:PRkernel} on the degree state space does not depend upon the random graph model in Section \ref{sec:graph}. Hence the derived lower bound of EI is useful for any degree correlation model.
\end{proof}

\section{Applications of Extremal Index in Network Sampling Processes}
\label{sec:applcn_ei}
This section provides several uses of EI to infer the sampled sequence. This emphasis that the analytical calculation and estimation of EI are practically relevant.

The limit of the point process of exceedances, $N_n(.)$, which counts the times, normalized by $n$, at which $\{X_i\}_{i=1}^n$ exceeds a threshold $u_n$ provides many applications of extremal index. A cluster is considered to be formed by the exceedances in a block of size $r_n$ ($r_n=o(n)$) in $n$ with cluster size $\xi_n=\sum_{i=1}^{r_n}1(X_i>u_n)$ when there is at least one exceedance within $r_n$. The point process $N_n$ converges weakly to a compound poisson process ($CP$) with rate $\theta \tau$ and i.i.d.\ distribution as the limiting distribution of cluster size, under condition \eqref{eq:CCDF_condn_theta} and a mixing condition, and the points of exceedances in $CP$ correspond to the clusters \cite[Section 10.3]{Beirlant}. We name this kind of clusters as blocks of exceedances.

The applications below require a choice of the threshold sequence $\{u_n\}$ satisfying \eqref{eq:CCDF_condn_theta}. For practical purposes, if a single threshold $u$ is demanded for the sampling budget $B$, we can fix $u=\max\{u_1,\ldots,u_B\}$.

The applications in this section are explained with the assumption that the sampled sequence is the sequence of node degrees. But the following techniques are very general and can be extended to any sampled sequence satisfying conditions $D(u_n)$ and $D''(u_n)$.

\subsection{Order statistics of the sampled degrees}
The order statistics $X_{n-k,n} $, $(n-k)$th maxima, is related to $N_n(.)$ and
 thus to $\theta$ by
\[\PP(X_{n-k,n} \leq u_n)=\PP(N_n((0,1])\leq k), \]
where we apply the result of convergence of $N_n$ to $CP$  \cite[Section 10.3.1]{Beirlant}.
\subsubsection{Distribution of Maxima}
The distribution of the maxima of the sampled degree sequences  can be derived as \eqref{eq:max_ei_reln} when $n\to\infty$.
Hence if the extremal index of the underlying process is known then from (\ref{eq:max_ei_reln}) one can approximate the $(1-\eta)$th quantile $x_{\eta}$ of the maximal degree $M_n$ as
\[\PP\{M_n\le x_{\eta}\}=F^{n\theta}(x_{\eta})=\PP^{n\theta}\{X_1\le x_{\eta}\}=1-\eta,\]
i.e.
\begin{equation}
\label{eq:quant_theta_genrl}
x_{\eta}\approx F^{-1}\left(\left(1-\eta\right)^{1/(n\theta)}\right).
\end{equation}
In other words, quantiles can be used to find the maxima of the degree sequence with certain probability. 

For a fixed certainty $\eta$, $x_\eta$ is proportional to $\theta$. Hence if the sampling procedures have same marginal distribution, with calculation of EI, it is possible to predict how much large values can be achieved. Lower EI indicates lower value for $x_{\eta}$ and higher represents
high
 $x_{\eta}$.

For the random walk example in Section \ref{subsec:theta_ex_RW} for the degree correlation model, with the use of (\ref{eq:quant_theta_genrl}), we get the  $(1-\eta)$th quantile of the maxima $M_n$
\[x_{\eta}\approx \mu+\sigma\left(\left(1-(1-\eta)^{1/(n \theta)}\right)^{-1/\gamma}-1\right) \]

The following example demonstrates the effect of neglecting correlations on the prediction of the largest degree node. The largest degree, with the assumption of Pareto distribution for the degree distribution, can be approximated as $KN^{1/\delta}$ with $K\approx 1$, $N$ as the number of nodes and $\gamma$ as the tail index of complementary distribution function of degrees \cite{Kostia2012}. For Twitter graph (recorded in 2012), $\delta=1.124$ for outdegree distribution and $N=537,523,432$ \cite{Maksym_2014}. This gives the largest degree prediction as $59,453,030$. But the actual largest out degree is $22,717,037$. This difference is because the analysis in \cite{Kostia2012} assumes i.i.d.\ samples and does not take into account the degree correlation. With the knowledge of EI, correlation can considered as in \eqref{eq:max_ei_reln}. In the following section, we derive an expression for such a case. 
\subsubsection{Estimation of largest degree when the marginals are Pareto distributed}
It is known that many social networks have the degree asymptotically distributed as Pareto. We find that in these cases, the marginal distribution of degrees of the random walk based methods also follow Pareto distribution (though we have derived only for the model with degree correlations among neighbors, see Section \ref{sec:deg_corlns})

\begin{claim}
For any stationary sequence with marginal distribution following Pareto distribution $\bar{F}(x)=Cx^{-\delta}$, the largest value is 
\[M_n \approx (n \theta)^{1/\delta} \Big(\frac{C}{\log 2} \Big)^{1/\delta} \]
\end{claim}
\begin{proof}
From extreme value theory \cite{Beirlant}, it is known that when $\{X_i, i\geq 1 \}$ are i.i.d.,
\begin{equation}
\label{eq:lar_deg_1}
\lim_{n \to \infty} \PP \left(\frac{M_n-b_n}{a_n} \leq x \right)=H_\gamma(x),
\end{equation}
where $H_{\gamma}(x)$ is the extreme value distribution with index $\gamma$ and  $\{a_n\}$ and $\{b_n\}$ are appropriately chosen deterministic sequences. When $\{X_i, i\geq 1 \}$ are stationary with EI $\theta$, the limiting distribution becomes $H'_{\gamma'}(x)$ and it differs from $H_{\gamma}(x)$ only through parameters. $H_{\gamma}(x)=\exp(-t(x))$ with $t(x)=\left(1+\left(\frac{x-\mu}{\sigma} \right) \gamma \right)^{-1/ \gamma}$. With the normalizing constants ($\mu=0$  and $\sigma=1$), $H'_{\gamma'}$ has the same shape as $H_\gamma$ with parameters $\gamma'=\gamma$, $\sigma'=\theta^{\gamma}$ and $\mu'=(\theta^{\gamma}-1)/\gamma$. 


For Pareto case, $\overline{F}(x)=Cx^{-\delta}$, $\gamma=1/ \delta$, $a_n=\gamma C^{\gamma} n^{\gamma}$ and $b_n=C^{\gamma}n^{\gamma}$. From \eqref{eq:lar_deg_1}, for large $n$,  $M_n$ is stochastically equivalent to $a_n \chi+b_n$, where $\chi$ is a random variable with distribution $H'_{\gamma'}$. It is observed in \cite{Kostia2012} that median of $\chi$ is an appropriate choice for the estimation of $M_n$. Median of $\chi=\mu'+\sigma'\left(\frac{(\log 2)^{-\gamma'}-1}{\gamma'} \right)= (\theta^{\gamma}(\log 2)^{-\gamma}-1)\gamma^{-1}$. Hence,
\begin{eqnarray}
M_n & \approx & a_n\left(\frac{\theta^\gamma (\log 2)^{-\gamma}}{\gamma}-1 \right)+b_n \nonumber\\
&= & (n \theta)^{1/\delta} \left( \frac{C}{\log 2}\right)^{1/ \delta} \nonumber
\end{eqnarray}
\end{proof}

\subsection{Relation to first hitting time and interpretations}
Extremal index also gives information about the first time $\{X_n\}$ hits $(u_n,\infty)$. Let $T_n$ be this time epoch. As $N_n$ converges to compound poisson process, it can be observed that $T_n/n$ is asymptotically an exponential random variable with rate $\theta \tau$, i.e., $\lim_{n\to \infty}\PP(T_n/n>x)=\exp(-\theta \tau x)$. Therefore $\lim_{n \to \infty}\EE(T_n/n)=1/(\theta \tau)$. Thus the more EI smaller, the more time it will take to hit the extreme levels as compared to independent sampling. This property can make use to compare different sampling procedures.
%

\subsection{Relation to mean cluster size}
If the conditions $D''(u_n)$ is satisfied along with $D(u_n)$, asymptotically, a run of the consecutive exceedances following an upcrossing is observed, i.e., $\{X_n\}$ crosses the threshold $u_n$ at a time epoch and stays above $u_n$ for some more time before crossing $u_n$ downwards and stays below it for some time until next upcrossing of $u_n$ happens. This is called cluster of exceedances and is more practically relevant than blocks of exceedances at the starting of this section and is shown in \cite{Leadbetter_1989} that these two definitions clusters are asymptotically equivalent resulting in similar cluster size distribution.
The expected value of cluster of exceedances converges to inverse of extremal index \cite[p.\ 384]{Beirlant}, i.e.,
\[\theta^{-1}=\lim_{n\to\infty}\sum_{j \geq 1}j\pi_n(j), \]
where $\{\pi_n(j),j\geq 1\}$ is the distribution of size of cluster of exceedances with $n$ samples. More details about cluster size distribution and its mean can be found in \cite{Markovich2014}.
%

\section{Estimation of Extremal Index and Numerical results}
\label{sec:simulations}
This section introduces two estimators for EI. Two types of networks are presented: synthetic correlated graph and real networks (Enron email network and DBLP network). For the synthetic graph, we compare the estimated EI to its theoretical value. For the real network, we calculate EI using the two estimators.

We take $\{X_i\}$ as the degree sequence and use RW, PR and RWJ as the sampling techniques. The methods mentioned in the following are general and are not specific to degree sequence or random walk technique.
\subsection{Empirical Copula based estimator}
We have tried different estimators for EI available in literature \cite{Beirlant,FerFer} and found that the idea of estimating copula and then finding value of its derivative at $(1,1)$ works without the need to choose and optimize several parameters found in other estimators. We assume that $\{X_i\}$ satisfies $D(u_n)$ and $D''(u_n)$ and we use \eqref{eq:theta_my_expn} for calculation of EI. Copula $C(u,v)$ is estimated empirically by
\[C_n(u,v)=\frac{1}{n}\sum_{k=1}^n \mathbb{I}\left(\frac{R_{i_k}^X}{n+1}\leq u, \frac{R_{i_k}^Y}{n+1}\leq v\right),\]
with $R_{i_k}^X$ indicates rank of the element $X_{i_k}$ in $\{X_{i_k},1\leq k \leq n\}$, and $Y_{i_k}=X_{i_k+1}$. The sequence $\{X_{i_k}\}$ is chosen from the original sequence $\{X_i \}$ in such a way that $X_{i_k}$ and $X_{i_{k+1}}$ are sufficiently apart to make them independent to certain extent. The large-sample distribution of $C_n (u, v)$ is normal and centered at copula $C(u, v)$. Now, to get $\theta$, we use linear least squares error fitting to find slope at 
$(1,1)$
 or use cubic spline interpolation for better results.
\subsection{Intervals Estimator}
This estimator does not assume any conditions on $\{X_i\}$, but has the parameter $u$ to choose appropriately. Let $N=\sum_{i=1}^{n} 1(X_i > u)$ be number of exceedances of u at time epochs $1 \leq S_1 < \ldots < S_N \leq n$ and let the interexceedance times are $T_i = S_{i+1}- S_{i}$. Then 
intervals 
estimator is defined as \cite[p.\ 391]{Beirlant},
\begin{equation*}
\hat{\theta}_n(u)=\Big\{
\begin{array}{ll}
\min(1,\hat{\theta}_n^1(u)) , \text{ if } \max{T_i : 1 \leq i \leq N -1} \leq 2,\\
\min(1,\hat{\theta}_n^2(u)) , \text{ if } \max{T_i : 1 \leq i \leq N -1} > 2,
\end{array}
\end{equation*}
where
\[\hat{\theta}_n^1(u)=\frac{2(\sum_{i=1}^{N-1}T_i)^2}{(N-1)\sum_{i=1}^{N-1}T_i^2},\]
and \[\hat{\theta}_n^2(u)=\frac{2(\sum_{i=1}^{N-1}(T_i-1))^2}{(N-1)\sum_{i=1}^{N-1}(T_i-1)(T_i-2)}.\]
We choose $u$ as $\delta$ percentage quantile thresholds, i.e., $\delta$ percentage of $\{X_i,1\leq i \leq n\}$ falls below $u$. The EI  is usually selected corresponding to the stability interval in the plot $(\theta, \delta)$.

\subsection{Synthetic graph} 
The simulations in the section follow the bivariate Pareto model and parameters introduced in \eqref{eq:biPareto}. We use the same set of parameters of Figure \ref{fig:trans_kernel_RW} and the graph is generated according to the technique in Section \ref{subsubsec:graph-genern}. 

For the RW case, Figure \ref{fig:ei_copula_estr} shows copula estimator, and theoretical copula based on the continuous distribution in \eqref{eq:biPareto} and is given by
\[C(u,u)=\left(1+2((1-u)^{-1/\gamma}-1)\right)^{-\gamma}+2u-1.\]
Though we take quantized values for degree sequence, it is found that the copula estimated matches with theoretical copula. The value of EI is then obtained after cubic interpolation and numerical differentiation at point $(1,1)$. For the theoretical copula, EI is $1-1/2^{\gamma}$, where $\gamma=1.2$.

\begin{figure}[!htb]
\centering
\includegraphics[trim = 16mm 0 19mm 5mm, clip=true, scale=0.23]{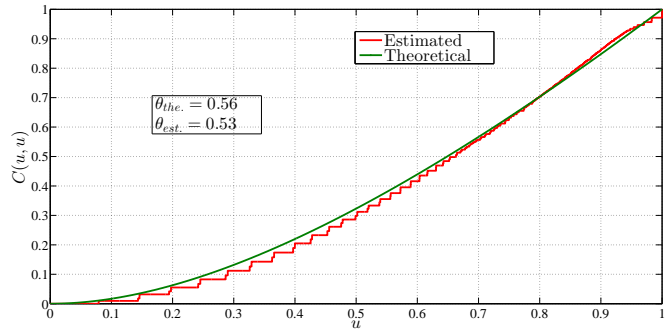}
\caption{Synthetic graph (RW sampling): Empirical and theoretical copulas}
\label{fig:ei_copula_estr}
\end{figure}

Figure \ref{fig:ei_ie_estr} displays the comparison between theoretical value of EI and Intervals estimate.
\begin{figure}[!htb]
\centering
\includegraphics[trim = 16mm 0 19mm 5mm, clip=true, scale=0.23]{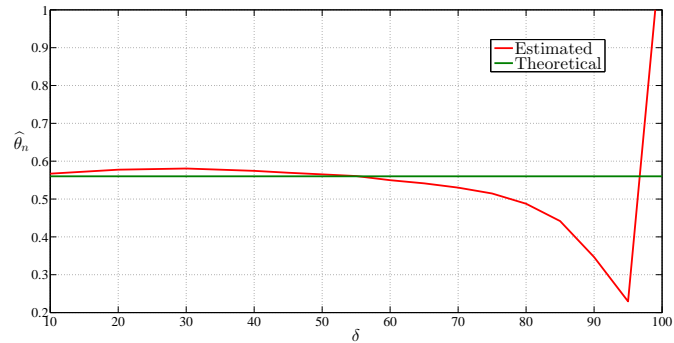}
\caption{Synthetic graph (RW sampling): Intervals estimate and theoretical value $\theta=0.56$ vs the percentage of quantile level, $\delta$}
\label{fig:ei_ie_estr}
\end{figure}

For the RWJ algorithm, Figure \ref{fig:ei_ie_estr_RWJ} shows the Intervals estimate and theoretical value for different $\alpha$. We used the expression \eqref{eq:EI-RWJ_2pareto} for theoretical calculation. The small difference in theory and simulation results is due to the assumption of continuous degrees in the analysis, but the practical usage requires quantized version. 

Figure \ref{fig:ei_ie_estr_PR} displays the Intervals estimate of EI with PR sampling. It can be seen that the lower bound proposed in Proposition \ref{prop:PR} gets tighter as $c$ decreases. 

\begin{figure}[!htb]
\centering
\includegraphics[trim = 16mm 0 19mm 5mm, clip=true, scale=0.23]{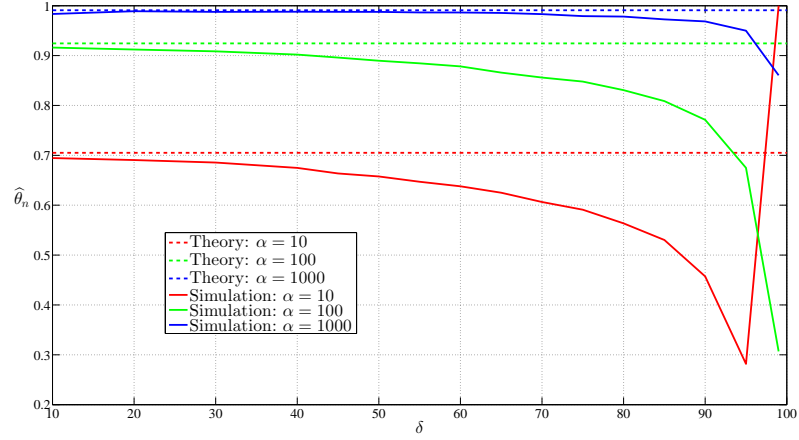}
\caption{Synthetic graph (RWJ sampling): Intervals estimate and theoretical value vs the percentage of quantile level $\delta$}
\label{fig:ei_ie_estr_RWJ}
\end{figure}

\begin{figure}[!htb]
\centering
\includegraphics[trim = 16mm 0 19mm 5mm, clip=true, scale=0.23]{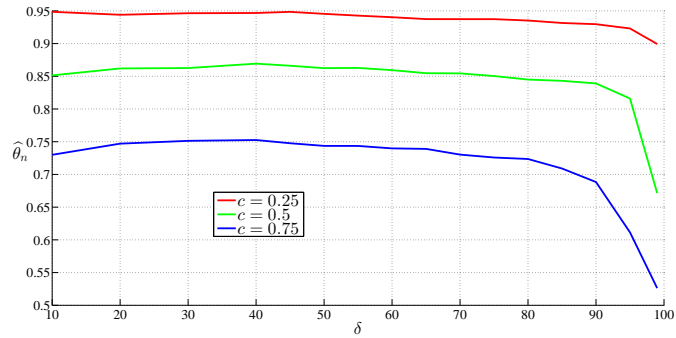}
\caption{Synthetic graph (PR sampling): Intervals estimate vs the percentage of quantile level $\delta$}
\label{fig:ei_ie_estr_PR}
\end{figure}

\subsubsection{Check of condition $D''$}
\label{subsub:check_d2_emp}
The mixing conditions $D(u_n)$ and $D''(u_n)$ need to be satisfied for using the theory in Section \ref{sec:calcn_EI}. Though Intervals estimator does not require them, these conditions will provide the existence of EI. Condition $D(u_n)$ works in this case as explained in previous sections and for $D''(u_n)$, we do the following empirical test. We collect samples for each of the techniques RW, PR and RWJ. Intervals are taken of duration $5, 10, 15$ and $20$ time samples. The ratio  of number of upcrossings to number of exceedances $r_{\textrm{up}}$ and ratio of number consecuitve exceedances to number of exceedances $r_{\textrm{cluster}}$ are calculated in Table \ref{tab:check_d2}. These proportions are averaged over $2000$ occurrences of each of these intervals and over all the different intervals. The statistics in the table indicates strong occurrence of  condition $D''(u_n)$. We have also observed that the changes in the parameters does not affect this inference. 
\begin{table}[!htb]
\begin{center}
\begin{tabular}[t]{|c|c|c|c|}\hline
 & $r_{\textrm{up}(\%)}$ & $r_{\textrm{cluster}(\%)}$\\ \hline
RW  & $4$ & $89$ \\ \hline
PR  & $7$ & $91$ \\ \hline
RWJ  & $5$ & $86$ \\ \hline
\end{tabular}
\caption{Test of Condition $D''$ in the synthetic graph}
\label{tab:check_d2}
\end{center}
\end{table}

\subsection{Real network}
We consider two real world networks: Enron email network and DBLP network. The data is collected from \cite{snap}. Both the networks satisfy the check for the condition $D''(u_n)$ reasonably well. 

For the RW sampling, Figure \ref{fig:cop_estr_enron_dblp} shows the bivariate copula estimated and mentions corresponding EI. Intervals estimator is presented in Figure \ref{fig:ie_estr_enron_dblp}. After observing plateaus in the plots, we took EI as $0.25$ and $0.2$ for DBLP and Enron email graphs, respectively.

\begin{figure}[!htb]
\centering
\includegraphics[trim = 16mm 0 19mm 5mm, clip=true, scale=0.23]{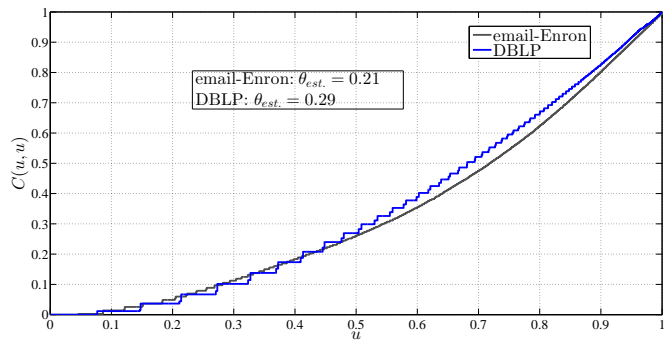}
\caption{Empirical copulas for email-Enron graph and DBLP graph (RW sampling)}
\label{fig:cop_estr_enron_dblp}
\end{figure}

\begin{figure}[!htb]
\centering
\includegraphics[trim = 16mm 0 19mm 5mm, clip=true, scale=0.23]{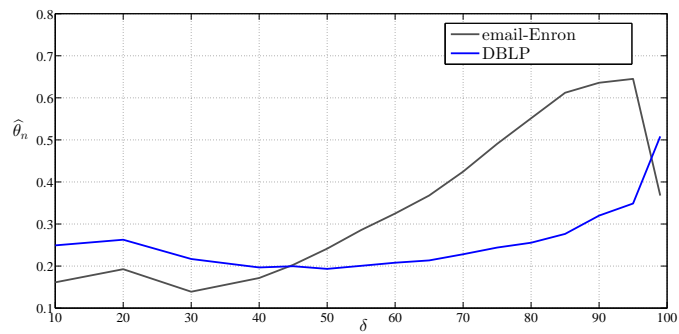}
\caption{email-Enron graph and DBLP graph (RW sampling): Intervals estimate vs the percentage of quantile level $\delta$.}
\label{fig:ie_estr_enron_dblp}
\end{figure}

In case of RWJ sampling, Figures \ref{fig:ie_estr_enron_RWJ} and \ref{fig:ie_estr_dblp_RWJ} present Intervals estimator for email-Enron and DBLP graphs respectively. Intervals estimate for PR sampling can be found in Figures \ref{fig:ie_estr_enron_PR} and \ref{fig:ie_estr_dblp_PR}. 

\begin{figure}[!htb]
\centering
\includegraphics[trim = 16mm 0 19mm 5mm, clip=true, scale=0.23]{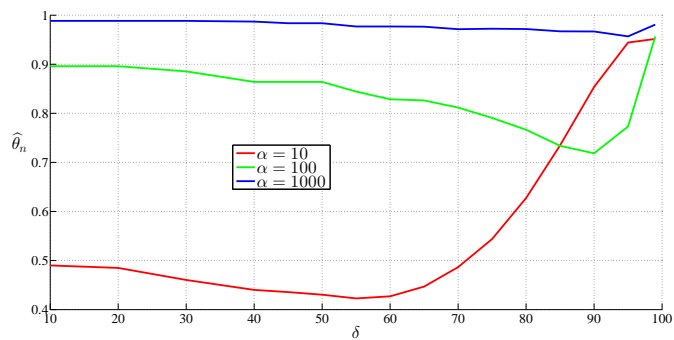}
\caption{email-Enron graph (RWJ sampling): Intervals estimate vs the percentage of quantile level $\delta$.}
\label{fig:ie_estr_enron_RWJ}
\end{figure}

\begin{figure}[!htb]
\centering
\includegraphics[trim = 16mm 0 19mm 5mm, clip=true, scale=0.23]{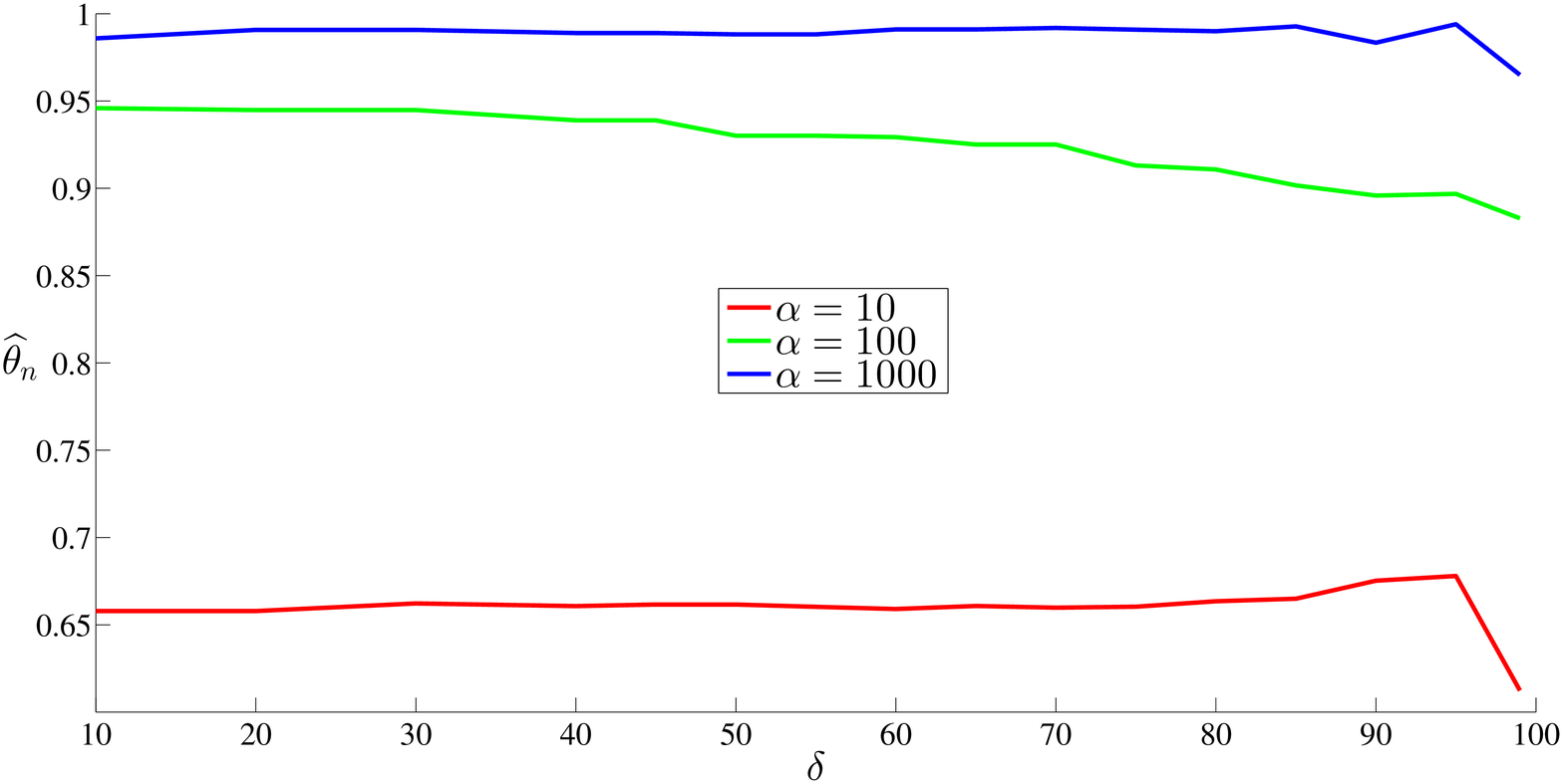}
\caption{DBLP graph (RWJ sampling): Intervals estimate vs the percentage of quantile level $\delta$.}
\label{fig:ie_estr_dblp_RWJ}
\end{figure}

\begin{figure}[!htb]
\centering
\includegraphics[trim = 16mm 0 19mm 5mm, clip=true, scale=0.23]{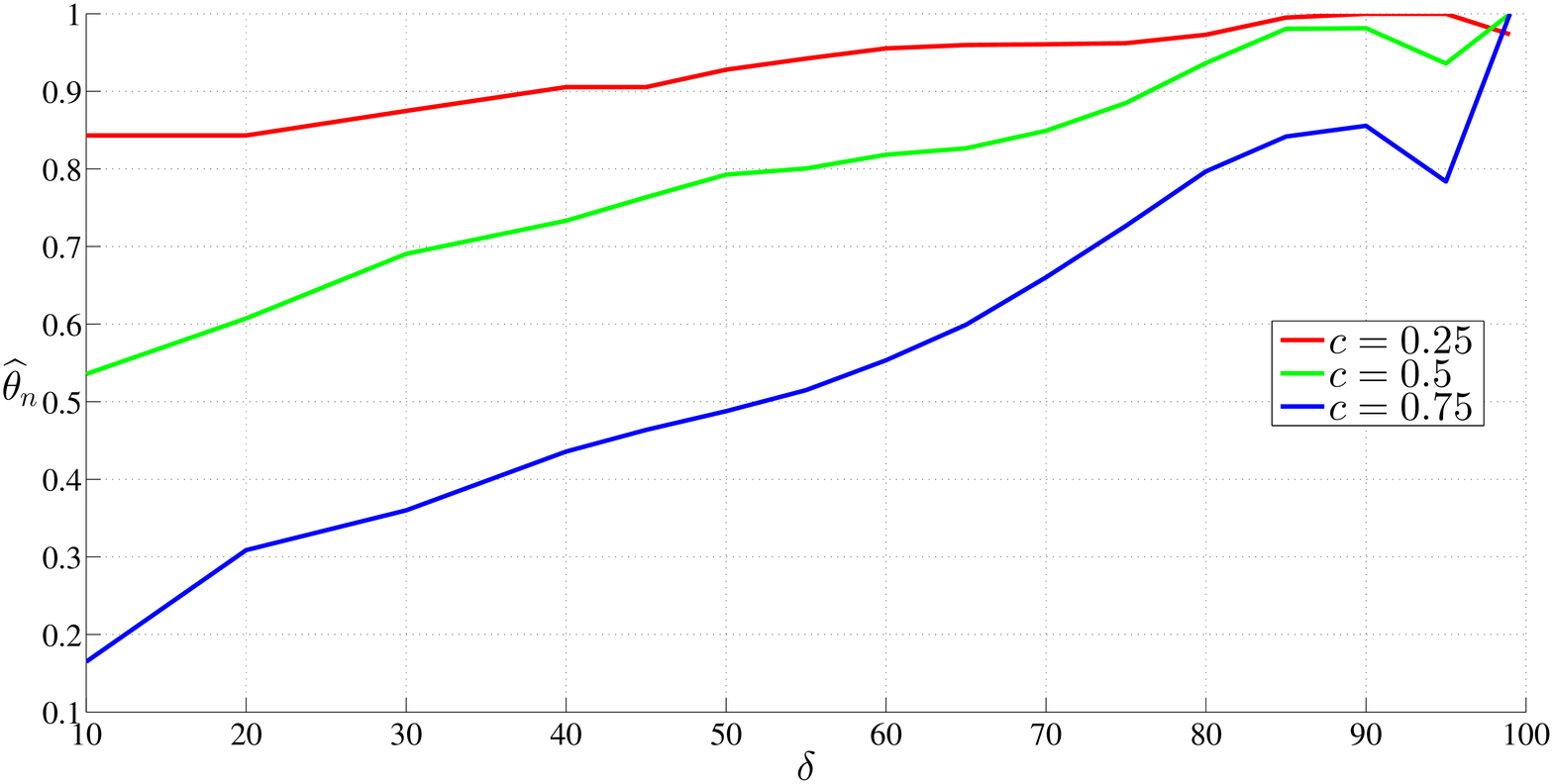}
\caption{email-Enron graph (PR sampling): Intervals estimate vs the percentage of quantile level $\delta$.}
\label{fig:ie_estr_enron_PR}
\end{figure}

\begin{figure}[!htb]
\centering
\includegraphics[trim = 16mm 0 19mm 5mm, clip=true, scale=0.23]{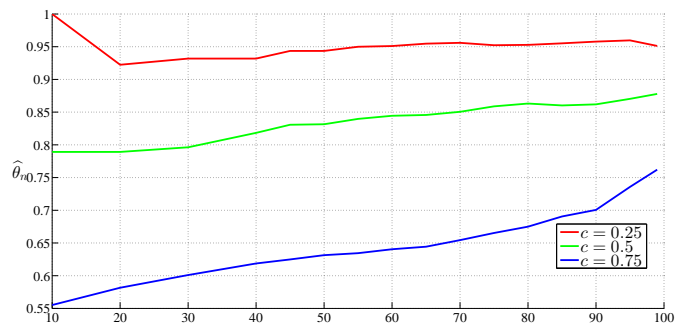}
\caption{DBLP graph (PR sampling): Intervals estimate vs the percentage of quantile level $\delta$.}
\label{fig:ie_estr_dblp_PR}
\end{figure}

\section{Conclusions}
\label{sec:conclu}
In this work, we have associated Extreme Value Theory of stationary sequences to sampling of large networks. We show that for any general stationary samples (function of node samples) meeting two mixing conditions, the knowledge of bivariate distribution or bivariate copula is sufficient to derive many of its extremal properties. The parameter extremal index (EI) encapsulates this relation. We relate EI to many relevant extremes in networks like order statistics, first hitting time, mean cluster size etc. In particular, we model correlation in degrees of adjacent nodes and examine samples from random walks on degree state space. Finally we have obtained estimates of EI for a synthetic graph with degree correlations and find a good match with the theory.  We also calculate EI for two real-world networks. In future, we plan to investigate the relation between assortativity coefficient and EI, and intends to study in detail the EI in real networks.

\section*{Endnotes}
\vspace{0.25cm}
\footnotesize
$^{(\textrm{a})}$ $F^{k}(.)$ indicates $k$th power of $F(.)$ throughout the paper except when $k=-1$ where it denotes the inverse function.\\

\noindent  $^{(\textrm{b})}$ $\sim$' stands for asymptotically equal, i.e. $f(x)\sim g(x)\Leftrightarrow f(x)/g(x)\rightarrow 1$ as $x\rightarrow a$, $x\in M$ where the functions $f(x)$ and $g(x)$ are defined on some set $M$ and $a$ is a limit point of $M$. 

$f(x)=o(g(x))$ means $\lim_{x\to a}f(x)/g(x)=0$.  Also $f(x)=O(g(x))$ indicates that there exist $\delta>0$ and $M>0$ such that $|f(x)| \le M |g(x)|$  for $|x - a| < \delta.$

\bibliographystyle{abbrv}
\bibliography{reference}


\end{document}